\documentclass[copyright]{eptcs}
\providecommand{\event}{T. Neary, D. Woods, A.K. Seda and N. Murphy (Eds.):\\ The Complexity of Simple Programs 2008.}
\providecommand{\volume}{1}
\providecommand{\anno}{2009}
\providecommand{\firstpage}{93}
\usepackage{breakurl}

\usepackage{mathrsfs}
\usepackage{amsmath}
\usepackage{amssymb}
\usepackage{amsthm}

\newtheorem{thm}{Theorem}

\newtheorem{lem}[thm]{Lemma}
\newtheorem{cor}[thm]{Corollary}
\newtheorem{exmp}[thm]{Example }
\newtheorem{defn}[thm]{Definition }

\newcommand{\lfam}{\mathscr{L}}
\newcommand{\subtext}[1]{\textnormal{\scriptsize #1}}

\newcommand{\tdha}[1]{\textrm{2DFA}(#1)}
\newcommand{\odha}[1]{\textrm{1DFA}(#1)}
\newcommand{\tnha}[1]{\textrm{2NFA}(#1)}
\newcommand{\onha}[1]{\textrm{1NFA}(#1)}

\newcommand{\tdidha}[1]{\textrm{2DiDFA}(#1)}
\newcommand{\odidha}[1]{\textrm{1DiDFA}(#1)}
\newcommand{\tdinha}[1]{\textrm{2DiNFA}(#1)}
\newcommand{\odinha}[1]{\textrm{1DiNFA}(#1)}

\newcommand{\tdpa}[1]{\textrm{2DPA}(#1)}

\newcommand{\tdbca}[1]{\textrm{2DBCA}(#1)}

\newcommand{\tdbcfa}[1]{\textrm{2DBCFA}(#1)}

\newcommand{\pcfa}{\textrm{PCFA}}
\newcommand{\dpcfa}{\textrm{DPCFA}}
\newcommand{\drpcfa}{\textrm{DRPCFA}}
\newcommand{\cpcfa}{\textrm{CPCFA}}
\newcommand{\dcpcfa}{\textrm{DCPCFA}}
\newcommand{\drcpcfa}{\textrm{DRCPCFA}}

\newcommand{\reg}{\textrm{REG}}

\newcommand{\leftend}{\mathord{\vartriangleright}}
\newcommand{\rightend}{\mathord{\vartriangleleft}}

\newcommand{\valc}{\textrm{VALC}}
\newcommand{\invalc}{\textrm{INVALC}}
\newcommand{\dollar}{\texttt{\$}}

\title{Multi-Head Finite Automata: Characterizations, Concepts and Open Problems}

\author{Markus Holzer
    \email{holzer@informatik.uni-giessen.de}
    \and
    Martin Kutrib
    \email{kutrib@informatik.uni-giessen.de}
    \and Andreas Malcher 
    \email{malcher@informatik.uni-giessen.de}
    \institute{Institut f\"ur Informatik, Universit\"at Giessen,\\
    Arndtstra{\ss}e~2, 35392 Giessen, Germany}
}
\def\titlerunning{Multi-Head Finite Automata: Characterizations, Concepts and Open Problems}
\def\authorrunning{M. Holzer, M. Kutrib and A. Malcher}

\begin{document}
\maketitle

\begin{abstract}
  Multi-head finite automata were introduced in~\cite{Rabin:1964:fadp}
  and~\cite{Rosenberg:1966:mhfa}. Since that time, a vast literature on
  computational and descriptional complexity issues on multi-head
  finite automata documenting the importance of these devices
  has been developed. Although multi-head finite automata are a simple concept,
  their computational behavior can be already very complex and leads
  to undecidable or even non-semi-decidable problems on these devices
  such as, for example, emptiness, finiteness, universality, equivalence,
  etc.  These strong negative results trigger the study of subclasses
  and alternative characterizations of multi-head finite automata
  for a better understanding of the nature of non-recursive trade-offs and, thus,
  the borderline between decidable and undecidable problems.  In
  the present paper, we tour a fragment of this literature.
\end{abstract}


\section{Introduction}
\label{sec:intro}

\noindent
Languages accepted by multi-tape or multi-head finite automata were
introduced in~\cite{Rabin:1964:fadp} and~\cite{Rosenberg:1966:mhfa}.
Since that time, many restrictions and generalizations of the original
models were investigated and studied (see, e.g.,~\cite{Wagner:1986:CC}). 
Two-way deterministic and nondeterministic
\emph{multi-head} finite automata are probably best known to
characterize the complexity classes of deterministic and nondeterministic
logarithmic space. In fact, in~\cite{Hartmanis:1972:ndscd} it was
shown that the question of the equality of deterministic and
nondeterministic logarithmic space is equivalent to the question of
whether every language accepted by some nondeterministic two-way
three-head finite automaton is accepted by some deterministic two-way
multi-head finite automaton.
Later this result was improved in \cite{sudborough:1975:tbccmhfa}.
Deterministic and nondeterministic one- and two-way multi-head
finite automata induce strict hierarchies of language families
with respect to the number of heads~\cite{monien:1980:twmhaola,Yao:1978:kpobtk}.  

Although multi-head finite automata are very simple devices, their
computational behavior is already highly complex. But what about the
size of multi-head finite automata opposed to their computational
power? Questions on the economics of description size were already
investigated in the early days of theoretical computer science and
build a cornerstone of descriptional complexity
theory~\cite{Meyer:1971:edagfs,Rabin:1959:fadp,Stearns:1967:cca}.  In
terms of descriptional complexity, a known upper bound for the
trade-off between different descriptional systems answers the
question, how succinctly a language can be represented by a descriptor
of one descriptional system compared to an equivalent description of
another descriptional system.  When dealing with more complicated
devices such as, for example, pushdown automata that accept regular
languages, a qualitative phenomenon revealed, namely that of
non-recursive trade-offs. There the gain in economy of description can
be arbitrary, that is, there are no recursive functions serving as
upper bounds for the trade-off. This phenomenon was first discovered
in~\cite{Meyer:1971:edagfs}, and in the meantime, a lot of deep and
interesting results on non-recursive trade-offs have been found for
powerful enough computational devices almost everywhere (see,
e.g.,~\cite{Goldstine:2002:dcmlr,Hartmanis:1980:sdrl:art,Hartmanis:1983:gsuslr,
  kutrib:2005:dphcp:art,kutrib:2005:pnrto:art,Schmidt:1977:sducfl,Valiant:1976:nsddl}).

Concerning descriptional complexity in connection with multi-head
finite automata, in~\cite{kutrib:2005:pnrto:art} it is shown that the
simulation of $\onha{2}$ by $\odha{k}$, for $k\geq 2$, causes
non-recursive trade-offs.  Similar results hold for two-way devices as
well~\cite{Kapoutsis:2005:kpotkdtonr}.  An immediate consequence from
the proofs of these results is that almost all of the aforementioned
elementary questions become undecidable---in fact they are even shown
to be not semi-decidable.  Furthermore, because of these non-recursive
trade-offs pumping lemmas and minimization algorithms for the automata
in question do not exist.
These strong negative results trigger the study of subclasses of
multi-head finite automata for a better understanding of the nature of
non-recursive trade-offs and, thus, the borderline between decidable
and undecidable problems. From the legion of possible research
directions we focus on three alternative and intermediate
computational models, namely (i) multi-head automata accepting bounded
languages, (ii) data-independent or oblivious multi-head finite
automata, and (iii) parallel communicating finite automata.  While the
former research on bounded languages dates back
to~\cite{Ginsburg:1966:MTCFL}, the latter two topics were recently
investigated in~\cite{holzer:1998:diddcmha:diss,holzer:2002:mhfadidd}
and~\cite{bordihn:2008:ccpcfa}. In fact, for some of these models,
some of the aforementioned elementary questions turn out to become
decidable, while for others undecidability remains.  At this point it
is worth mentioning, that recently it was shown that even stateless
one-way multi-head finite automata have a non-decidable emptiness
problem~\cite{Ibarra:2007:smahep:tr}. In fact, these devices are the
most simple one, since they have one internal state only.

In the present paper we tour a fragment of the literature on
computational and descriptional complexity issues of multi-head finite
automata. It obviously lacks completeness, as one falls short of
exhausting the large selection of multi-head finite automata related
problems considered in the literature.  
We give our view of what
constitute the most recent interesting links to the considered problem
areas.

\section{Multi-Head Finite Automata}
\label{sec:multi-head-automata}

\noindent
We denote the set of non-negative integers by $\mathbb{N}$. 
We write $\Sigma^*$ for the set of all words over the finite alphabet $\Sigma$.
The empty word is denoted by~$\lambda$, and 
$\Sigma^+ = \Sigma^* \setminus \{\lambda\}$. The reversal of a word $w$ is
denoted by~$w^R$ and for the length of $w$ we write~$|w|$. 
We use $\subseteq$ for inclusions and~$\subset$ for strict inclusions. 
We write~$2^{S}$ for the powerset of a set $S$.

Let $k \geq 1$ be a natural number. A two-way $k$-head finite
automaton is a finite automaton having a single read-only input tape whose
inscription is the input word in between two endmarkers. The $k$ heads of the
automaton can move freely on the tape but not beyond the endmarkers.
A formal definition is:
\begin{defn} A \emph{nondeterministic two-way $k$-head finite
automaton $(\tnha{k})$} is a system\linebreak $\langle S,A,k,\delta,
\leftend,\rightend,s_0,F\rangle$, where
$S$ is the finite set of \emph{internal states}, 
$A$ is the set of \emph{input symbols}, 
$k\geq 1$ is the \emph{number of heads}, 
$\leftend \notin A$ and
          $\rightend\notin A$ are the \emph{left and right endmarkers}, 
$s_0 \in S$ is the \emph{initial state}, 
$F\subseteq S$ is the set of \emph{accepting states}, an 
$\delta$ is the partial transition function mapping 
$S \times (A \cup \{\leftend, \rightend\})^k$ into the subsets of
$S \times \{-1,0,1\}^k$, where $1$ means to move the head one
square to the right, $-1$ means to move it one square to the left,
and $0$ means to keep the head on the current square. Whenever  
$(s', (d_1,\dots,d_k)) \in \delta (s, (a_1, \dots, a_k))$ is defined, then $d_i
\in \{0,1\}$ if $a_i = \leftend$, and $d_i \in \{-1,0\}$ if
$a_i = \rightend$, $1\leq i \leq k$.
\end{defn}

A $\tnha{k}$ starts with all of its heads on the first square of the tape. It
halts when the transition function is not defined for the current situation. 
A \emph{configuration} of a $\tnha{k}$ $M=\langle S,A,k,\delta,\leftend,
\rightend, s_0,F\rangle$ at some time~$t$ with $t \geq 0$ is a
triple $c_t=(w,s,p)$, where $w$ is the input, $s \in S$ is the current state, 
and $p= (p_1,\dots, p_k) \in \{0,\dots,|w|+1\}^k$ gives the current head
positions.  If a position $p_i$ is $0$, then head $i$ is scanning the
symbol~$\leftend$, if it is $n+1$, then the head is
scanning the symbol $\rightend$. The initial configuration for 
input $w$ is set to $(w, s_0, (1,\dots,1))$.  During its course of 
computation, $M$ runs through a
sequence of configurations. One step from a configuration to its successor
configuration is denoted by~$\vdash$.
Let $w= a_1 a_2\ldots a_n$ be the input, $a_0= \leftend$, and
$a_{n+1}= \rightend$, then we set $ (w,s,(p_1,\dots, p_k))
\vdash$ \linebreak $(w,s',(p_1 +d_1,\dots,p_k +d_k)) $ if and only if 
$(s', (d_1,\dots, d_k))\in \delta(s,(a_{p_1},\dots,a_{p_k}))$.  
As usual we define the
reflexive, transitive closure of $\vdash$ by $\vdash^*$.  Note, that due to
the restriction of the transition function, the heads cannot move beyond the
endmarkers.

The language accepted by a $\tnha{k}$ is precisely the set of words
$w$ such that there is some computation beginning with $\leftend
w\rightend$ on the input tape and ending with the $\tnha{k}$ halting
in an accepting state, i.e.,
$
L(M) = \{\,w \in A^* \mid (w, s_0, (1,\dots,1)) \vdash^*
(w,s,(p_1,\dots,p_k)), s \in F,
\mbox{ and } M \mbox{ halts in } (w,s, (p_1,\dots,p_k))\,\}
$.
If in any case $\delta$ is either undefined or a singleton, then 
the  $k$-head finite automaton is said to be \emph{deterministic}.
Deterministic two-way $k$-head finite automata are denoted by $\tdha{k}$. 
If the heads never move to the left, then the  $k$-head finite 
automaton is said to be \emph{one-way}. Nondeterministic and
deterministic one-way $k$-head finite automata are denoted by 
$\onha{k}$ and $\odha{k}$.
The family of all languages accepted by a device of some type $X$ is denoted by
$\lfam(X)$.  

The power of multi-head finite automata is well studied in the
literature. For one-head machines we obtain 
a characterization of the family of regular languages $\reg$. 
A natural question is to what extent the computational power depends
on the number of heads. For (one-way) automata the proper inclusion 
$\lfam(\onha{1})\subset\lfam(\odha{2})$ is evident. An early result 
is the inclusion which separates the next level, that is, 
$\lfam(\odha{2})\subset\lfam(\odha{3})$~\cite{ibarra:1974:thv2hfa}.
The breakthrough occurred in~\cite{Yao:1978:kpobtk}, where it was shown
that the language
$
L_n = \{\,w_1 \dollar w_2 \dollar \ldots \dollar w_{2n} \mid w_i \in \{a,b\}^*
\mbox{ and } w_i =w_{2n+1-i} \mbox{ for } 1\leq i \leq n\,\}
$
can be used to separate the computational power of automata with~$k+1$ heads
from those with~$k$ heads in the one-way setting:
\begin{thm} 
Let $k\geq 1$. 
Then $\lfam(\odha{k})\subset\lfam(\odha{k+1})$ and 
$\lfam(\onha{k})\subset\lfam(\onha{k+1})$.
\end{thm}

By exploiting the same language, the computational power of nondeterministic
classes could be separated from the power of deterministic classes. To this
end, for any $n$ the complement of $L_n$ was shown to be accepted by some one-way two-head
nondeterministic finite automaton. Since the deterministic language families 
$\lfam(\odha{k})$ are closed under complementation, it follows that 
the inclusions $\lfam(\odha{k})\subset\lfam(\onha{k})$ are
proper, for all $k\geq 2$.
In order to compare one- and two-way multi-head finite automata classes,
let $L=\{\,w\mid w\in\{a,b\}^* \mbox{ and } w=w^R\,\}$ be the 
mirror language.  It is well known that the mirror language is not accepted
by any $\onha{k}$, but its complement belongs to $\lfam(\onha{2})$.
The next corollary summarizes the inclusions.
\begin{sloppypar}
\begin{cor}
  Let $k\geq 2$. 
Then
  $\lfam(\odha{k})\subset\lfam(\tdha{k})$,
  $\lfam(\odha{k})\subset\lfam(\onha{k})$, and 
  $\lfam(\onha{k})\subset\lfam(\tnha{k})$.
\end{cor}
\end{sloppypar}
From the complexity point of view, the
two-way case is the more interesting one, since there is the following strong relation
to the computational complexity classes $\mathsf{L}=\mathsf{DSPACE}(\log(n))$ and
$\mathsf{NL}=\mathsf{NSPACE}(\log(n))$~\cite{Hartmanis:1972:ndscd}.
\begin{thm}
  $\mathsf{L} = \bigcup_{k\geq 1}\lfam(\tdha{k})$ and $\mathsf{NL} =
  \bigcup_{k\geq 1}\lfam(\tnha{k})$.
\end{thm}

Concerning a head hierarchy of \emph{two-way} multi-head finite automata, 
in~\cite{monien:1980:twmhaola} it was shown that $k+1$ heads are better than~$k$.
Moreover, the witness languages are \emph{unary}.
\begin{thm}
Let $k\geq 1$. 
Then there are unary languages that 
show the inclusions
$\lfam(\tdha{k})\subset\lfam(\tdha{k+1})$ and 
$\lfam(\tnha{k})\subset\lfam(\tnha{k+1})$.
\end{thm}

Whether nondeterminism is better than determinism in the two-way setting,
is an open problem. 
In fact, in~\cite{sudborough:1975:tbccmhfa} it was shown that
the equality for at least one $k\geq 2$ implies $\mathsf{L}=\mathsf{NL}$. 
More generally, 
$
\mathsf{L}=\mathsf{NL} \mbox{ if and only if }
\lfam(\onha{2})\subseteq \bigcup_{k\geq 1}\lfam(\tdha{k})
$
was shown in~\cite{sudborough:1975:tbccmhfa} .
Due to a wide range of relations between several types of finite
automata with different resources, the results and open problems for $k$-head
finite automata apply in a similar way for other types.
Here, we mention deterministic two-way finite automata with~$k$
pebbles ($\tdpa{k}$), with $k$ linearly bounded counters ($\tdbca{k}$), and
with~$k$ linearly bounded counters with full-test ($\tdbcfa{k}$).  In order to
adapt the results, we present the hierarchy 
$
\lfam(\tdha{k}) \subseteq \lfam(\tdpa{k}) \subseteq \lfam(\tdbca{k})
\subseteq \lfam(\tdbcfa{k}) \subseteq \lfam(\tdha{k+1})
$
which has been shown in several famous papers,
e.g., \cite{Morita:1977:ccnbcamra,Petersen:1995:ash,Ritchie:1972:lrma,Sugata:1977:laraca}.

\section{Descriptional Complexity}

\noindent
It is natural to investigate the succinctness of the representations
of formal languages by different models.  For example, it is well known
that two-way and one-way finite automata are equivalent.
Recently, the problem of the costs in terms of states for these simulations
was
solved in~\cite{kapoutsis:2005:rbdnfa} by establishing a tight bound
of~$\binom{2n}{n+1}$ for the simulation of two-way deterministic as well as
nondeterministic finite automata by one-way nondeterministic finite automata.
In the same paper tight bounds of $n(n^n-(n-1)^n)$ and 
$\sum_{i=0}^{n-1}\sum_{j=0}^{n-1}\binom{n}{i}\binom{n}{j}(2^i-1)^j$
are shown for two-way deterministic and two-way nondeterministic finite
automata simulations by one-way deterministic finite automata.
Nevertheless, some challenging problems of
finite automata are still open.  An important example is the question
of how many states are sufficient and necessary to simulate $\tnha{1}$
with $\tdha{1}$; we refer
to~\cite{Hromkovic:2003:nvdtwfa,Leung:2001:tlbssa,Sakoda:1978:NST,Sipser:1980:lpssa}
for further reading.
All trade-offs mentioned so far are bounded by recursive
functions. But, for example, there is no recursive function which bounds the
savings in descriptional complexity between deterministic and unambiguous
pushdown automata~\cite{Valiant:1976:nsddl}.  
A survey of the phenomenon of non-recursive trade-offs
is~\cite{kutrib:2005:pnrto:art}.
Here we ask for the descriptional complexity of $k$-head finite automata. 
How succinctly can a language be 
presented by a $k$-head finite automaton 
compared with the presentation by a nondeterministic pushdown automaton, or
by a $(k+1)$-head finite automaton?
\begin{defn} 
A \emph{descriptional system} $E$ is a recursive set of finite descriptors,
where each descriptor $D \in E$
describes a formal language $L(D)$, and 
there exists an effective procedure to convert
$D$ into a Turing machine that decides (semi-decides)
the language $L(D)$, if $L(D)$ is recursive (recursively enumerable).
The \emph{family of languages represented (or described) by some descriptional
system}~$E$ is $L(E) =\{\,L(D) \mid D \in E\,\}$. 
For every language $L$, we define $E(L) =\{\,D \in E \mid L(D) =L\,\}$.

A \emph{complexity (size) measure} for
$E$ is a total, recursive function $c: E \to \mathbb{N}$, such that for
any alphabet~$A$, the set of descriptors in $E$ describing languages
over $A$ is 
recursively enumerable in order of increasing size, and
does not contain infinitely many descriptors of the same size.

Let $E_1$ and $E_2$ be two descriptional systems, and $c$ be a complexity
measure for $E_1$ and~$E_2$. A function $f: \mathbb{N} \to \mathbb{N}$, 
with $f(n) \geq n$, is said
to be an \emph{upper bound} for the increase in complexity when changing from a
minimal description in~$E_1$ to an equivalent minimal description in $E_2$, 
if for all $L \in L(E_1) \cap L(E_2)$ we have  
$
\min \{\,c(D) \mid D \in E_2(L) \,\} \leq 
f(\min \{ \,c(D) \mid D \in E_1(L)\,\}).
$
If there is no recursive upper bound, the \emph{trade-off is said to be 
non-recursive}. 
\end{defn}

In connection with multi-head finite automata the question of determining the trade-offs
between the levels of the head hierarchies arises immediately.
Recently, the problem whether these trade-offs are non-recursive 
has been solved for two-way devices in the affirmative~\cite{Kapoutsis:2005:kpotkdtonr}
(cf.~also \cite{kutrib:2005:dphcp:art}):
\begin{thm}
Let $k\geq 1$. 
Then the trade-off between 
deterministic (nondeterministic) two-way $(k+1)$-head finite automata and 
deterministic (nondeterministic) two-way $k$-head finite automata is 
non-recursive. Moreover, the results hold also for automata accepting unary languages.
\end{thm}

But how to prove non-recursive trade-offs? Roughly speaking, most of 
the proofs appearing in the literature are
basically relying on one of two different schemes. 
One of these techniques is due to
Hartmanis~\cite{Hartmanis:1980:sdrl:art,Hartmanis:1983:gsuslr}.
Next, we present a slightly generalized and unified form of this 
technique~\cite{kutrib:2005:dphcp:art,kutrib:2005:pnrto:art}. 
\begin{thm}\label{theo:technique1}
Let $E_1$ and $E_2$ be two descriptional systems for
recursive languages.  If there exists a descriptional system $E_3$
such that, given an arbitrary ${M} \in E_3$,
(i)~there exists an effective procedure to construct a descriptor in 
$E_1$ for some language~$L_{{M}}$, 
(ii)~$L_{{M}}$ has a descriptor in~$E_2$ if and only if
$L({M})$ does not have a property~$P$, and
(iii)~property~$P$
is not semi-decidable for languages with descriptors in~$E_3$,
then the trade-off between~$E_1$ and~$E_2$ is non-recursive.
\end{thm}
In order to apply Theorem~\ref{theo:technique1} one needs a descriptional 
system $E_3$ with
appropriate problems that are not even semi-decidable. 
An immediate descriptional system is the set of Turing
machines for which only trivial problems are decidable and a lot of problems
are not semi-decidable. 
Basically, we consider \emph{valid computations of Turing
machines}~\cite{Hartmanis:1967:cfltmc}. Roughly 
speaking, these are histories of accepting Turing machine computations. 
It suffices to consider deterministic Turing machines with one single tape
and one single read-write head. 
Without loss of generality and for technical reasons, we assume that the Turing
machines can halt only after an odd number of moves, accept by halting, make
at least three moves, and cannot print a blank. 
A valid computation is a string built from a sequence of configurations
passed through during an accepting computation. 
Let $S$ be the state set of some Turing machine~$M$,
where $s_0$ is the initial state, $T \cap S = \emptyset$ is the tape alphabet
containing the blank symbol, $A\subset T$ is the input alphabet, and
$F \subseteq S$ is the set of accepting states.  Then a configuration of
$M$ can be written as a word of the form $T^*ST^*$ such that
$t_1\ldots t_i s t_{i+1} \ldots t_n$ is used to express that~$M$ is
in state $s$, scanning tape symbol $t_{i+1}$, and $t_1$ to $t_n$ is the
support of the tape inscription.

For the purpose of the following, valid computations
$\mathrm{VALC}(M)$ are now defined to be the set of words
$\dollar w_1 \dollar w_2\dollar \ldots \dollar w_{2n}\dollar$,
where $\dollar \notin T\cup S$, $w_i
\in T^*ST^*$ are configurations of $M$, $w_1$ is an initial
configuration of the form~$s_0A^*$, $w_{2n}$ is an accepting configuration of the
form $T^*FT^*$, and $w_{i+1}$ is the successor configuration of $w_i$,
for $1\leq
i < 2n$.
The set of \emph{invalid computations} $\invalc(M)$
is the complement of $\valc({M})$ with respect to the alphabet
$\{\dollar\} \cup T \cup S$. 

Now we complement the results on the trade-offs between the levels of the
head hierarchy for two-way automata. We show non-recursive trade-offs between 
the levels of the head hierarchies of deterministic as well as nondeterministic one-way
devices. Moreover, non-recursive trade-offs are shown
between nondeterministic two-head and deterministic $k$-head automata.
\begin{exmp}\label{exa:twoheads-accept-valid}
Let ${M}$ be a Turing machine. Then an $\odha{2}$ ${M'}$ can be
constructed that accepts $\valc({M})$.
One task of ${M'}$ is to verify the correct form of the input, that is, whether it is
of the form 
$\dollar s_0 A^* \dollar T^* S T^*\dollar \ldots \dollar T^* S T^*\dollar
T^* FT^* \dollar$.
This task means to verify whether the input belongs to a
regular set and can be done in parallel to the second task. 

The second task is to verify for each two adjacent subwords whether the second
one represents the successor configuration of the first one.  We show the
construction for $w_i \dollar w_{i+1}$.  Starting with the first
head on the first symbol of~$w_i$ and the second head on the first symbol of
$w_{i+1}$, automaton ${M'}$ compares the subwords symbolwise by moving
the heads to the right. Turing machine ${M}$ has three possibilities
to move its head.  So, $w_i = t_1 \ldots t_i s t_{i+1} \ldots t_n$ goes to
$t_1 \ldots t_i s' t_{i+1}' \ldots t_n$, to $t_1 \ldots s' t_i t_{i+1}' \ldots
t_n$, or to $t_1 \ldots t_i t_{i+1}' s' \ldots t_n$.  Each of the three
possibilities can be detected by ${M'}$.  Furthermore, ${M'}$
can verify whether the differences between $w_i$ and $w_{i+1}$ are due to a
possible application of the transition function of ${M}$.  Finally, 
the first head is moved on the first symbol of $w_{i+1}$, and the second head 
is moved on the first symbol of $w_{i+2}$ to start the verification of $w_{i+1}$ and
$w_{i+2}$.
\end{exmp}

\begin{thm}
Let $k \geq 2$.
The trade-off between $\odha{k}$ and
nondeterministic pushdown automata is non-recursive.
\end{thm}

\begin{proof}
In order to apply Theorem~\ref{theo:technique1}, let $E_3$ be the set
of Turing machines.  For every ${M} \in E_3$, define
$L_{{M}}$ to be $\valc({M})$.  So, $L_{{M}}$
is accepted by some $\odha{k}$. In~\cite{Hartmanis:1967:cfltmc}
it was shown that $\valc({M})$ is context free if and only if $L({M})$
is finite. Since infiniteness is not semi-decidable for Turing machines,
all conditions of Theorem~\ref{theo:technique1} are satisfied and
the assertion follows.
\end{proof}
\begin{cor}\label{cor:start-hierarchy}
Let $k \geq 2$. \!The trade-off between $\odha{k}$ and
$\onha{1}$ is non-recursive.
\end{cor}
In the following we exploit the basic hierarchy theorem shown 
in~\cite{Yao:1978:kpobtk}, and recall that the language
$
L_n = \{\,w_1 \dollar w_2 \dollar \ldots \dollar w_{2n} \mid w_i \in \{a,b\}^*
\mbox{ and } w_i =w_{2n+1-i} \mbox{ for } 1\leq i \leq n\,\}
$
is accepted by some deterministic or nondeterministic one-way $k$-head finite
automaton if and only if $n \leq \binom{k}{2}$.
Now we extend the language in order to meet our purposes. 
Basically, the idea is to keep the structure of the language but to build 
the subwords~$w_i$ over an alphabet of pairs, 
that is, 
$w_i={u_1\atop v_1}{u_2\atop v_2}{\ldots\atop \ldots}{u_m\atop v_m}$,
where the $u_j$ and $v_j$ are symbols such that the upper parts of the
subwords are words over~$\{a,b\}$ as in~$L_n$. 
The lower parts are valid computations of some given Turing machine ${M}$.
Let $W_{{M}} = \{\, {u\atop v} \mid u \in \{a,b\}^*, v\in \valc({M}),
|u|=|v|\,\}$. Then $L_{n,{M}}$ is defined to be 
$\{\,w_1\dollar w_2\dollar\ldots\dollar w_{2n} \mid w_i \in W_{{M}}$
and $w_i =w_{2n+1-i}$ for $1\leq i \leq n\,\}$.
\begin{lem}\label{lem:positive-hierarchy}
Let $k \geq 2$ 
and ${M}$ be some Turing
machine. Then a $\odha{k+1}$ can be constructed that accepts 
$L_{\binom{k}{2}+1, {M}}$.
\end{lem}

\begin{proof}
Let $n= \binom{k}{2}+1$. 
We sketch the construction. At first, two of the heads, say the first and\linebreak
$(k+1)$st one, are used to verify that the lower parts of the subwords
$w_1, w_2,
\ldots, w_n$ do belong to $\valc ({M})$ 
(cf.~Example~\ref{exa:twoheads-accept-valid}). 
At the end of this task, both heads are positioned at the beginning of
$w_{n+1}$. 
Next, the first head is moved to the beginning of $w_{2n+1-(k-1)}$ while
the remaining $k-1$ heads are moved to the beginnings of $w_1, w_2,\dots,
w_{k-1}$, respectively. Now the $k-1$ words $w_{2n+1-(k-1)}$ to $w_{2n}$ are
successively compared with the words $w_{k-1}$ to~$w_1$. 
Next, the $k-1$ heads are moved to the beginning of $w_k$, and the same
procedure is applied inductively to verify $w_k \dollar w_{k+1}\dollar\ldots\dollar
w_{2n+1-k}$.

After $k-1$ repetitions 
$\sum_{i=k-1}^1 i = \frac{k^2-k}{2}=\binom{k}{2}$ pairs have been compared,
that is, the pairs $w_i$ and $w_{2n+1-i}$, for $1\leq i\leq n-1$. So, one of the heads that
compared $w_{n-1}$ and $w_{n+2}$ is positioned at the beginning of
$w_n$. Since head $k+1$ is still positioned at the beginning of $w_{n+1}$, the
remaining pair $w_n$ and $w_{n+1}$ can be compared.
\end{proof}

It is a straightforward adaption of the hierarchy result of~\cite{Yao:1978:kpobtk}
to prove that the language $L_{\binom{k}{2}+1, {M}}$ is not accepted by any
$\onha{k}$ if $L({M})$ is infinite. 

Now we can apply Theorem~\ref{theo:technique1} in order to show the non-recursive trade-offs
between any two levels of the deterministic or nondeterministic head hierarchy.
\begin{thm}
Let $k\geq 1$. 
The trade-offs between $\odha{k+1}$
and $\odha{k}$, between $\onha{k+1}$ and $\onha{k}$, and
between $\odha{k+1}$ and $\onha{k}$
are non-recursive.
\end{thm}
\begin{proof}
For $k=1$, the theorem has been shown by Corollary~\ref{cor:start-hierarchy}. 
So, let $k\geq2$. We apply Theorem~\ref{theo:technique1} by setting $E_3$ to be the set of
Turing machines and $L_{{M}} = L_{\binom{k}{2}+1, {M}}$. As property $P$
we choose infiniteness.
Lemma~\ref{lem:positive-hierarchy} shows that the language 
$L_{{M}}$ is accepted
by some $\odha{k+1}$ and, therefore, by some $\onha{k+1}$. 
On the other hand, language~$L_{{M}}$ is not
accepted by any $\onha{k}$ and,
therefore, not accepted by any $\odha{k}$ if $L({M})$ is
infinite. Conversely, if $L({M})$ is finite, the set $\valc ({M})$
is finite. This implies that $L_{\binom{k}{2}+1, {M}} = L_{{M}}$ is
finite and, thus, is accepted by some $\odha{k}$ and, therefore, by some
$\onha{k}$. 
\end{proof}
The next question asks for the trade-offs between nondeterministic and
deterministic automata. Clearly, the trade-off between $\onha{1}$ and
$\odha{1}$ is recursive. But following an idea 
in~\cite{Yao:1978:kpobtk} which separates 
nondeterminism from determinism, we can show non-recursive trade-offs on 
every level $k\geq2$ of the hierarchy as follows.
\begin{thm}
Let $k\geq 2$. 
Then the trade-off between 
$\onha{2}$ and $\odha{k}$ and, thus, between 
$\onha{k}$ and
$\odha{k}$ is non-recursive.
\end{thm}

Another consequence of the fact that the set of valid computations is accepted
by $k$-head finite automata is that many of their properties are not even
semi-decidable. We can transfer the results from Turing machines.
\begin{thm}\label{theo:non-semidecidable}
Let $k\geq 2$. 
Then the problems of 
emptiness, finiteness, infiniteness,
universality, inclusion, equivalence, regularity, and context-freedom
are not semi-decidable for
$\lfam(\odha{k})$, $\lfam(\onha{k})$, $\lfam(\tdha{k})$, and $\lfam(\tnha{k})$. 
\end{thm}
\begin{thm}
Let $k\geq 2$. 
Then any language family whose word problem is semi-decidable and
that effectively contains the language families $\lfam(\odha{k})$, $\lfam(\onha{k})$,
$\lfam(\tdha{k})$, or $\lfam(\tnha{k})$ 
does not possess a pumping lemma.
\end{thm}
Finally, note that no minimization algorithm for the
aforementioned devices exists, since otherwise emptiness becomes
decidable, which contradicts 
Theorem~\ref{theo:non-semidecidable}.
\begin{thm}
There is no minimization algorithm converting some
$\odha{k}$, $\onha{k}$, $\tdha{k}$, or $\tnha{k}$,
for $k\geq 2$, to an equivalent automaton of the same
type with a minimal number of states.
\end{thm}

\section{Alternative and Intermediate Computational Models}

\noindent
The existence of non-recursive trade-offs and the undecidability of many
decidability problems for multi-head finite automata is at least disconcerting
from an applied perspective.
So, the question arises under which assumptions
undecidable questions become decidable. Here we focus on
three different alternative and intermediate
computational models, namely (i) multi-head automata accepting bounded
languages, (ii) data-independent or oblivious multi-head finite automata, and
(iii) parallel communicating finite automata. 

\subsection{Multi-Head Finite Automata Accepting Bounded Languages}

\noindent
If we impose the
structural restriction of ``boundedness,'' then for context-free grammars and
pushdown automata it is known that the trade-offs become recursive and
decidability questions become decidable~\cite{malcher:2007:dcobcl}.
In this section, we summarize corresponding results for multi-head finite
automata accepting bounded languages. 

\begin{defn}
Let $A=\{a_1,a_2,\dots,a_n\}$.
A language $L\subseteq A^*$ is said to be \emph{letter-bounded} or
\emph{bounded} if $L \subseteq
a_1^*a_2^* \ldots a_n^*$.
A subset $P \subseteq \mathbb{N}^n$ is said to be a \emph{linear set} if there
exist $\alpha_0, \alpha_1, \ldots, \alpha_m$ in $\mathbb{N}^n$ such that
$
P=\{\,\beta \mid \beta = \alpha_0 + \sum_{i=1}^m a_i \alpha_i,\mbox{ where 
$a_i \geq 0$,} \mbox{ for } 1 \le i \le m\,\}
$
and~$P$ is said to be \emph{semilinear} if it is the finite
union of linear sets.
\end{defn}

In order to relate bounded languages and semilinear sets we introduce the following
notation. 
For an alphabet $A=\{a_1,a_2,\dots,a_n\}$ the \emph{Parikh mapping}
$\Psi:\Sigma^*\to \mathbb{N}^n$ is defined by
$\Psi(w)=$\linebreak$(|w|_{a_1},|w|_{a_2},\dots, |w|_{a_n})$, where
$|w|_{a_i}$ denotes the number of occurrences of~$a_i$ in the word~$w$.
We say that a language is \emph{semilinear} if its Parikh image is a semilinear set.
In~\cite{Parikh:1966:ocfl} a fundamental result concerning the distribution of
symbols in the words of a context-free language has been shown. It says that for any
context-free language $L$, the Parikh image $\Psi(L)=\{\,\Psi(w)\mid w\in L\,\}$ is
semilinear.
Semilinear sets have many appealing properties. For example, they are closed under
union, intersection, and complementation
\cite{Ginsburg:1966:MTCFL}. Furthermore, bounded
semilinear languages have nice decidability properties, since the questions of emptiness, universality, 
finiteness, infiniteness, inclusion, and equivalence are decidable (cf.~\cite{Ginsburg:1966:MTCFL}).
The main connection between bounded languages and multi-head finite automata
is the following result~\cite{ibarra:1974:nslsbrmhpda,Sudborough:1974:brmfal}.

\begin{lem}\label{lemma:char}
Let $k\geq 1$ and $L$ be a bounded language
accepted by a $\tnha{k}$ which performs a constant number of head
reversals. Then $\Psi(L)$ is a semilinear set. 
\end{lem}

The above characterization implies immediately that there are no one-way multi-head
finite automata which accept the non-semilinear language $L_1=\{\,a^{2^n} \mid n \ge
1\,\}$. 
The situation changes for unbounded languages. 
For example, the non-semilinear language $L_2=\{\,aba^2ba^3b \ldots a^nb \mid n \ge 1\,\}$ is
accepted by some $\odha{2}$ by checking that the size of adjacent $a$-blocks
is always increasing by one. 
Due to the relation to semilinear sets and their positive decidability results
we obtain the borderline between decidability and undecidability as follows.

\begin{thm}
Let $k\geq 2$.
The problems of
emptiness, universality, finiteness, infiniteness, inclusion, and equivalence are undecidable for
the language families $\lfam(\odha{k})$, $\lfam(\onha{k})$, $\lfam(\tdha{k})$, and
$\lfam(\tnha{k})$,
if the automata are allowed to perform at most a constant number of head
reversals.
The problems become decidable for bounded languages.
\end{thm}

We next summarize results concerning decidability questions on models
which are generalized or restricted versions of multi-head finite
automata. Let us start with nondeterministic two-way multi-head pushdown
automata~\cite{ibarra:1974:nslsbrmhpda} which are nondeterministic two-way multi-head finite
automata augmented by a pushdown store.
In fact, the characterization in Lemma~\ref{lemma:char} has been shown 
in~\cite{ibarra:1974:nslsbrmhpda} for the stronger model of nondeterministic 
two-way multi-head pushdown automata. So, 
for $k\geq 1$ let a bounded language $L$ be accepted by some nondeterministic 
two-way $k$-head pushdown automaton
which performs a constant number of head reversals. Then $\Psi(L)$ is a
semilinear set.

Thus, the positive decidability results obviously hold also for such automata.
As is the case for multi-head finite automata, we loose positive decidability
and semilinearity if we drop the condition of boundedness. We may consider
again the language~$L_2$ which is clearly accepted by a one-way two-head
pushdown automaton as well.

In the general model of multi-head finite automata the moves of the heads are
only depending on the input, and a head cannot feel the presence of another
head at the same position at the same time. If this property is added to each
head, we come to the model of one-way multi-head finite automata having \emph{sensing heads}
(see~\cite{ibarra:1974:nslsbrmhpda}). This model is stronger since, for
example, non-semilinear language\linebreak $\{\,a^{n^2} \mid n \ge 1\,\}$ can be accepted by a
deterministic one-way finite automaton having three sensing heads,
whereas with non-sensing heads
only semilinear languages are accepted in the unary case
(see~Lemma~\ref{lemma:char}). 
On the other hand, it is shown in~\cite{ibarra:1974:nslsbrmhpda} that every
bounded language which is accepted by some nondeterministic one-way finite
automaton having two sensing heads can be also accepted by some
nondeterministic one-way pushdown automaton with two heads, and thus is
semilinear. 
\begin{thm}
The problems of emptiness, universality, finiteness, infiniteness, inclusion,
and equivalence are undecidable for deterministic one-way finite automata with
at least three sensing heads.
The aforementioned problems become decidable for nondeterministic one-way finite automata
with two sensing heads that accept a bounded language.
\end{thm}

The computational capacity of two-way finite automata with sensing heads and
bounds on the number of input reversals is investigated
in~\cite{hromkovic:1985:fatnmawrnr}.
We next consider a restricted version of one-way multi-head finite automata
which allows only one designated head to read and distinguish input symbols
whereas the remaining heads cannot distinguish input symbols, that is, the
remaining heads get only the information whether they are reading an input
symbol or the right endmarker. These automata are called
\emph{partially blind} multi-head finite automata~\cite{ibarra:2006:opbmfa}.
In this model the decidability results are much better, since it can be shown
that every language, not necessarily bounded, which is accepted by a partially
blind nondeterministic one-way multi-head finite automaton has a semilinear Parikh
map. In detail, we have the following theorem.

\begin{thm}[\cite{ibarra:2006:opbmfa}]
(1) The problems of emptiness, finiteness, and infiniteness are decidable
for partially blind nondeterministic one-way multi-head finite automata.
Universality is undecidable and, thus, the problems of inclusion and
equivalence are undecidable as well.
(2) The language family accepted by partially blind deterministic
one-way multi-head finite automata is closed under
complementation. Thus, the problems of universality, inclusion, and
equivalence are decidable.
(3) The problems of emptiness, universality, finiteness, infiniteness,
inclusion, and equivalence are decidable
for partially blind nondeterministic one-way multi-head finite automata
that accept a bounded language.
\end{thm}

\subsection{Data-Independent Multi-Head Finite Automata}

\noindent
The notion of oblivious Turing programs was introduced in~\cite{Paterson:1974:ioaolm}.
Recently, obliviousness was investigated in relation with computational models from classical
automata theory as, for example, multi-head finite automata in~\cite{holzer:2002:mhfadidd}. 
\begin{defn}
  A $k$-head finite automaton~$M$ is \emph{data-in\-de\-pen\-dent} or \emph{oblivious} if
  the position of every \emph{input-head}~$i$ after
  step~$t$ in the computation on input~$w$ is a function
  $f_M({|w|},i,t)$ that only depends on~$i$, $t$,~and~$|w|$.
\end{defn}

We denote deterministic (nondeterministic) one-way and two-way
data-in\-de\-pen\-dent $k$-head finite automata by
$\odidha{k}$ and $\tdidha{k}$ ($\odinha{k}$ and $\tdinha{k}$). 
In~\cite{Petersen:1998:hhofapac} it was shown that
data-independence is no restriction for multi-head one-counter
automata, multi-head non-erasing stack automata, and multi-head stack
automata. The same is obviously true for one-head finite automata.

Interestingly, one can show that for data-independent $k$-head finite automata
determinism is as powerful as nondeterminism. To this end, observe
that a data-independent automaton with~$k$ heads has, on inputs of length~$n$,
only one possible input-head trajectory during a computation.
So, the only nondeterminism left in the
$k$-head nondeterministic finite automata is the way in which the next state is
chosen. Therefore, a power-set construction shows that even data-independent
determinism may simulate data-independent nondeterminism. 
Thus, we have the following theorem.
\begin{sloppypar}
\begin{thm}\label{equivalence-for-di}
  Let $k\geq 1$.
Then 
  $\lfam(\odidha{k})=\lfam(\odinha{k})$ and
moreover $\lfam(\tdidha{k})=\lfam(\tdinha{k})$.
\end{thm}
\end{sloppypar}

But why are data-independent multi-head finite automata that interesting?
Similarly as ordinary multi-head finite automata characterize logarithmic
space bounded Turing machine computations~\cite{Hartmanis:1972:ndscd}, 
data-independent multi-head finite automata capture
the parallel complexity class $\mathsf{NC^1}$, which 
is the set of problems accepted by log-space uniform
circuit families of logarithmic depth, polynomial size, with AND- and
OR-gates of bounded fan-in.  We have $\mathsf{NC^1}\subseteq\mathsf{L}$.
For the characterization the following observation
is useful: Every deterministic multi-head finite automaton that works on unary inputs
is already data-independent by definition. Therefore,
$
\mathsf{L^u}\subseteq\bigcup_{k\geq 1}\lfam(\tdidha{k})
$,
where $\mathsf{L^u}$ denotes the set of all unary languages from
$\mathsf{L}$. However, because of the trajectory argument, we do not
know whether the inclusion 
$\mathsf{NL^u}\subseteq\bigcup_{k\geq 1}\tdinha{k}$ holds. A positive
answer would imply $\mathsf{NL^u}\subseteq\mathsf{L^u}$ by 
Theorem~\ref{equivalence-for-di}. This, in turn, would lead 
to a positive answer of the \emph{linear bounded automaton (LBA) problem}
by translational methods~\cite{kuroda:1964:cllba}. 
In~\cite{holzer:2002:mhfadidd} the following theorem was shown.
\begin{thm}\label{NCe-by-DiDFA}
  $\mathsf{NC^1}=\bigcup_{k\geq 1}\lfam(\tdidha{k})$.
\end{thm}
Since there exists a log-space complete language in $\lfam(\odha{2})$,
we immediately obtain 
$\mathsf{NC^1}=\mathsf{L}$ if and only if
$\lfam(\odha{2})\subseteq\bigcup_{k\geq 1}\lfam(\tdidha{k})$.
Similarly, $\mathsf{NC^1}=\mathsf{NL}$ if and only if 
$\lfam(\onha{2})\subseteq\bigcup_{k\geq 1}\lfam(\tdidha{k})$.

What is known about the head hierarchies induced by data-independent 
multi-head finite automata?
For two-way data-independent multi-head finite automata one can profit from known
results~\cite{monien:1980:twmhaola}. Since for every $k\geq 1$
there are unary languages that may serve as witnesses
for the inclusions 
$\lfam(\tdha{k})\subset\lfam(\tdha{k+1})$ and  
$\lfam(\tnha{k})\subset\lfam(\tnha{k+1})$, we obtain the following
corollary.
\begin{cor}
Let $k\geq 1$. Then 
$\lfam(\tdidha{k})\subset\lfam(\tdidha{k+1})$.
\end{cor}
 
The remaining inclusions $\lfam(\tdidha{k})\subseteq\lfam(\tdha{k})$, for $k\geq
2$, for two-way multi-head finite automata are related to the question of
whether $\mathsf{NC^1}=\mathsf{L}$. This parallels the issue
of whether the inclusion $\lfam(\tdha{k})\subseteq\lfam(\tnha{k})$ is proper for
$k\geq 2$. For the relationship between data-independent and
data-dependent multi-head finite automata, we find  that
$\lfam(\tdidha{k})=\lfam(\tdha{k})$, for some $k\geq 2$, 
implies $\mathsf{NC^1}=\mathsf{L}$.
Moreover, it was shown in~\cite{holzer:1998:diddcmha:diss} that the head hierarchy for
one-way data-independent multi-head finite automata is strict. Obviously,
$\reg=\lfam(\odidha{1})$ and, hence, $\reg\subset\lfam(\tdidha{2})$.
By adapting the head hierarchy result of~\cite{Yao:1978:kpobtk}
to data-independent automata the next theorem is shown.
\begin{thm}\label{thm:bound better than Yao Rivest}
Let $k\geq 1$. Then
  $\lfam(\odidha{k})\subset\lfam(\odidha{\frac{k\cdot(k+1)}{2}+3})$.
\end{thm}

Thus, we have an infinite proper hierarchy with respect to the number
of heads, but the bound obtained in
Theorem~\ref{thm:bound better than Yao Rivest} is not very good, especially for
small values of~$k$. In fact, a separation for the first four levels
was obtained in~\cite{holzer:1998:diddcmha:diss}, using the language
$
  L_{F,n} = \{\,a^{i\cdot F(2)}\dollar a^{i\cdot F(3)}\dollar\ldots
  \dollar a^{i\cdot F(k+1)}\mid i\geq 1\,\},
$
  where $F(j)$ is the $j$th \emph{Fibonacci number}. It holds
$L_{F,n}\in\lfam(\odidha{k})$, if and only if $n\leq \frac{k\cdot (k-1)}{2}+1$.
\begin{thm}
$\lfam(\odidha{1})\subset \lfam(\tdidha{2})\subset
  \lfam(\tdidha{3})\subset\lfam(\tdidha{4})$.
\end{thm}

Whether the one-way hierarchy for data-independent multi-head finite automata is strict in
the sense that $k+1$ heads are better than~$k$,
is an open problem. Concerning the remaining
open inclusions the following is known.
\begin{thm}\label{thm:one-way subset two-way} 
Let $k\geq 2$. Then 
$\lfam(\odidha{k})\subset\lfam(\tdidha{k})$.
\end{thm}

Moreover, by the copy language as witness we find that one-way
data-in\-de\-pen\-dent 
languages are a proper subset of two-way
data-independent multi-head finite automata languages.  
\begin{cor}  
$\bigcup_{k\geq 1}\lfam(\odidha{k})\subset\bigcup_{k\geq 1}\lfam(\tdidha{k})=\mathsf{NC^1}$.
\end{cor}
We close this subsection by mentioning some open problems for further research:
(1) Determine the bounds of the conversions of one-head data-independent finite automata
into one-head data-dependent deterministic, nondeterministic, and alternating finite automata and
\textit{vice versa}.  (2) Consider decidability and complexity
questions such as equivalence, non-emptiness, etc.\ for $k$-head data-independent
finite automata.  Finally, the most interesting point for research might be (3)
the one-way $k$-head hierarchy for data-independent finite automata. Is it a
strict hierarchy, in the sense that $k+1$ heads 
are better than~$k$?

\subsection{Parallel Communicating Finite Automata}

\noindent
In this section we will focus on parallel communicating finite automata systems which
were introduced in \cite{Martinvide:2002:pfascs}. 
Basically, a parallel communicating finite automata system of degree~$k$ is 
a device of~$k$ finite automata working in parallel with each other on a common one-way 
read-only input tape and being synchronized according to a universal clock.
The~$k$ automata communicate on request by states, that is, 
when some automaton enters a distinguished query state $q_i$, it is set to the current 
state of automaton~$A_i$. Concerning the next state of the sender $A_i$,
we distinguish two modes. In \emph{non-returning} mode the sender remains in
its current state whereas in \emph{returning} mode the sender is set
to its initial state. Moreover, we distinguish whether all automata are
allowed to request communications, or whether there is just one master
allowed to request communications. The latter types are called 
\emph{centralized}. 
\begin{defn}
A \emph{nondeterministic parallel communicating finite automata system of
degree~$k$} $(\pcfa(k))$ is a construct 
$\mathcal{A}=\langle \Sigma, A_1, A_2, \dots, A_k, Q, \rightend\rangle$, 
where~$\Sigma$ is the set of \emph{input symbols}, 
each\linebreak $A_i=\langle S_i, \Sigma,
\delta_i, s_{0,i}, F_i\rangle$, for $1\leq i\leq k$, is a 
\emph{nondeterministic finite automaton} with state set $S_i$, 
initial state $s_{0,i}\in S_i$, set of
accepting states $F_i\subseteq S_i$, and transition function 
$\delta_i: S_i \times (\Sigma \cup \{\lambda,\rightend\}) \to 2^{S_i}$,
 $Q=\{q_1, q_2, \dots, q_k\}\subseteq \bigcup_{1 \leq i \leq k} S_i$
 is the set of \emph{query states}, and
$\rightend\notin \Sigma$ is the \emph{end-of-input symbol}.
\end{defn}

The automata $A_1, A_2, \dots, A_k$ are called \emph{components} of the
system $\mathcal{A}$.
A \emph{configuration} \linebreak $(s_1, x_1, s_2, x_2, \dots, s_k, x_k)$ of 
$\mathcal{A}$ represents the current states $s_i$ 
as well as the still unread parts $x_i$ of the tape inscription of all 
components $1\leq i\leq k$. System $\mathcal{A}$ starts with all of 
its components scanning the first square of the tape in their initial states. 
For input word $w\in \Sigma^*$, the initial configuration is
$(s_{0,1}, w\rightend, s_{0,2}, w\rightend, \dots, s_{0,k}, w\rightend)$.
Each step can consist of two phases. 
In a first phase, all components are in non-query states and
perform an ordinary (non-communicating) step independently.
The second phase is the communication phase during which components 
in query states receive the requested states as long as the sender 
is not in a query state itself. This process is repeated until all requests 
are resolved, if possible. If the requests are cyclic, no successor configuration exists.
As mentioned above, we distinguish \emph{non-returning} communication, 
that is, the sender remains in its current state, and \emph{returning} 
communication, that is, the sender is reset to its initial state.
For the first phase, we define the successor configuration relation~$\vdash$ by 
$
(s_1, a_1y_1, s_2, a_2y_2, \dots, s_k, a_ky_k) \vdash 
(p_1, z_1, p_2, z_2, \dots, p_k, z_k),
$
if $Q\cap \{s_1,s_2,\dots,s_k\} = \emptyset$, 
$a_i\in \Sigma\cup\{\lambda,\rightend\}$, $p_i\in \delta_i(s_i,a_i)$, 
and $z_i=\rightend$ for $a_i =\rightend$ and $z_i=y_i$ otherwise,
for $1\leq i \leq k$.
For non-returning communication in the second phase, we set 
$
(s_1, x_1, s_2, x_2, \dots, s_k, x_k) \vdash 
(p_1, x_1, p_2, x_2, \dots, p_k, x_k), 
$
if, for all $1\leq i\leq k$ such that $s_i = q_j$ and $s_j\notin Q$, we
have $p_i=s_j$, and $p_r=s_r$ for all the other $r$, for $1\leq r\leq k$.
Alternatively, for returning communication in the second phase, 
we set 
$
(s_1, x_1, s_2, x_2, \dots, s_k, x_k) \vdash 
(p_1, x_1, p_2, x_2, \dots, p_k, x_k), 
$
if, for all $1\leq i\leq k$ such that $s_i = q_j$ and $s_j\notin Q$, we
have $p_i=s_j$, $p_j=s_{0,j}$, and $p_r=s_r$ for all the other $r$, $1\leq r\leq k$.

A computation \emph{halts} when the successor configuration is 
not defined for the current situation. In particular, this may happen
when cyclic communication requests appear, or when the transition
function of one component is not defined. 
The language $L(\mathcal{A})$ accepted by a $\pcfa(k)$ $\mathcal{A}$
is precisely the set of words~$w$ such that there is some computation beginning with 
$w\rightend$ on the input tape and halting with at least one component
having an undefined  transition function and
being in an accepting state. Let $\vdash^*$ denote the reflexive and
transitive closure of~$\vdash$ and set
$L(\mathcal{A}) = \{\,w\in \Sigma^* \mid 
(s_{0,1}, w\rightend, s_{0,2}, w\rightend, \dots, s_{0,k}, w\rightend)
\vdash^*$ $
(p_1, a_1y_1, p_2, a_2y_2, \dots, p_k, a_ky_k),$
such that $p_i\in F_i$ and $\delta_i(p_i,a_i)$  is undefined,
for some $1\leq i\leq k\,\}$.

If all components $A_i$ are deterministic finite automata, then
the whole system is called \emph{deterministic}, and we add the prefix D to denote
it. The absence or presence of an R in the type of the system denotes
whether it works in \emph{non-returning} or \emph{returning} mode, respectively. 
Finally, if there 
is just one component, say $A_1$, that is allowed to query for states, that
is, $S_i\cap Q=\emptyset$, for $2\leq i\leq k$, then the system is said to be
\emph{centralized}. We denote centralized systems by a C. Whenever the degree
is missing we mean systems of arbitrary degree.
In order to clarify our notation we give an example.  
\begin{exmp}
We consider the
language $\{\,w \texttt{\$} w \mid w \in \{a,b\}^+\,\}$ and show that it can be
accepted by a $\dcpcfa$ with two components. 
Thus, all types of systems of parallel communicating finite automata accept more than 
regular languages.
The rough idea of the construction is that in every time step the master
component queries the non-master component, and the non-master component reads
an input symbol. When the non-master component has read the separating 
symbol~$\texttt{\$}$, which is notified to the master with the help of primed
states, then the master component
starts to compare its input symbol with the information from the non-master
component. If all symbols up to $\texttt{\$}$ match, the input is accepted and
in all other cases rejected. The precise construction is given through the
following transition functions.

{\small
\[
\begin{array}{rcl@{\hspace{4em}}rcl@{\hspace{4em}}rcl}
 \delta_1(s_{0,1},\lambda) &=& q_2 & \delta_1(s_a,\lambda) &=& q_2&
 \delta_1(s_b,\lambda) &=& q_2\\ \delta_1(s_\mathtt{\$},\lambda) &=& q_2 &
 \delta_1(s'_a,a) &=& q_2& \delta_1(s'_b,b) &=& q_2\\
 \delta_1(s_{\rightend},\texttt{\$}) &=& accept & && & &&\\[1mm]
 \delta_2(s_{0,2},a) &=& s_a & \delta_2(s_{0,2},b) &=& s_b & &&\\
 \delta_2(s_a,a) &=& s_a & \delta_2(s_a,b) &=& s_b & \delta_2(s_a,\texttt{\$})
 &=& s_\mathtt{\$}\\ \delta_2(s_b,a) &=& s_a & \delta_2(s_b,b) &=& s_b &
 \delta_2(s_b,\texttt{\$}) &=& s_\mathtt{\$}\\ \delta_2(s_\mathtt{\$},a) &=&
 s'_a & \delta_2(s_\mathtt{\$},b) &=& s'_b& \delta_2(s_{\rightend},\rightend) &=& s_{\rightend}\\ 
\delta_2(s'_a,a) &=& s'_a & \delta_2(s'_a,b) &=& s'_b & \delta_2(s'_a,\rightend) &=& s_{\rightend}\\
 \delta_2(s'_b,a) &=& s'_a & \delta_2(s'_b,b) &=& s'_b & \delta_2(s'_b,\rightend) &=& s_{\rightend}\\
\end{array}
\]
}
\end{exmp}

For nondeterministic non-centralized devices it is shown in~\cite{Choudhary:2007:rnrpcnfae}
that returning parallel communicating finite automata systems are neither weaker nor
stronger than non-returning ones. The question whether the same equivalence is true in the 
deterministic case was answered in the affirmative
in~\cite{bordihn:2008:ccpcfa}.
To this end, the so-called \emph{cycling-token-method} is introduced.
Basically, the main problem of that method was to break
the synchronization at the beginning. Otherwise, when some component $A_{i+1}$
requests the state of $A_i$ and, thus,~$A_i$ reenters its initial state, then 
$A_i$ will request the state of~$A_{i-1}$ and so on. But these 
cascading communication requests would destroy necessary
information. 
The next lemma is shown by applying the cycling-token-method.
\begin{lem}\label{lem:rpcdfa-includes-pcdfa}
Let $k\geq 1$. 
Then $\lfam(\drpcfa(k))$ 
includes $\lfam(\dpcfa(k))$.
\end{lem}
One of the fundamental results obtained in~\cite{Martinvide:2002:pfascs}
is the characterization of the computational power of (unrestricted)
parallel communicating finite automata systems by multi-head finite automata.
\begin{thm}\label{theo:rpcdfa=pcdfa=kdha}
Let $k\geq 1$. 
Then the families $\lfam(\drpcfa(k))$, $\lfam(\dpcfa(k))$,
and $\lfam(\odha{k})$ are equal.
\end{thm}
\begin{proof}
It remains to be shown that, for all $k\geq 1$, 
the family $\lfam(\odha{k})$ includes $\lfam(\drpcfa(k))$.
Given some $\drpcfa(k)$ $\mathcal{A}$, basically, the idea of simulating it by
a deterministic one-way $k$-head finite automaton $A'$ is to
track all current states of the components of $\mathcal{A}$ in the current
state of $A'$ (see~\cite{Martinvide:2002:pfascs}).
\end{proof}

Comparing deterministic centralized systems with
non-central\-ized systems we obtain for centralized systems the surprising result that
the returning mode is not weaker than the non-returning mode. Let
$
L_\subtext{rc}= \{\, uc^xv\texttt{\$}uv \mid u,v\in\{a,b\}^*, x\geq 0\,\}.
$
\begin{thm}\label{theo:drcpcfa-not-dpcfa}
The language $L_\subtext{rc}$ belongs to the family $\lfam(\drcpcfa)$ (and thus to
$\lfam(\drpcfa)=\lfam(\dpcfa)$), but not to $\lfam(\dcpcfa)$.
\end{thm}
\begin{cor}\label{cor_gencap1}
$\lfam(\dcpcfa) \subset \lfam(\dpcfa)=\lfam(\drpcfa)$.
\end{cor}

In order to show that nondeterministic centralized systems are strictly more powerful
than their deterministic variants we consider the
language $L_\subtext{mi}= \{\,ww^R \mid w \in \{a,b,c\}^+\,\}$ and
show that its complement belongs to $\lfam(\cpcfa)$, but does not belong to
$\lfam(\dpcfa)$.
\begin{cor}\label{cor_gencap2}
$\lfam(\dcpcfa) \subset \lfam(\cpcfa)$ and 
$\lfam(\dpcfa) \subset \lfam(\pcfa)$.
\end{cor}
Finally, we compare the classes under consideration with some well-known 
language families.
\begin{lem} \label{lemma_dcsl}
The family $\lfam(\pcfa)$ is strictly included in~$\textsf{NL}$, hence, in the family of deterministic context-sensitive languages. 
\end{lem} 
\begin{proof}
Since $L=\{\,ww^R \mid w \in \{a,b,c\}^+\,\}$ can be accepted by some 
$\tnha{2}$, language~$L$ belongs to $\mathsf{NL}$.
On the other hand, language $L$ does not belong to $\lfam(\onha{k})$, for all $k
\geq 1$.
\end{proof}
\begin{lem}\label{lemma_dlin}
All language classes accepted by parallel communicating finite automata
systems are incomparable to the class of (deterministic) (linear) context-free languages.
\end{lem} 
\begin{proof}
The language $\{\,w \texttt{\$} w^R \mid w \in \{a,b\}^+\,\}$ of marked
palindromes of even lengths is deterministic linear context free, but is not
accepted by any $\onha{k}$, for all $k \geq 1$. Thus, 
the language $\{\,w \texttt{\$} w^R \mid w \in \{a,b\}^+\,\}$ does not belong to
$\lfam(\pcfa)$. 
Conversely,
the set $\{\,w \texttt{\$} w \mid w \in \{a,b\}^+\,\}$ belongs to
$\lfam(\drcpcfa)$ as well as to $\lfam(\dcpcfa)$, but is not
context free.
\end{proof}
\begin{lem}\label{lemma_crl}
All language classes accepted by parallel communicating finite automata
systems are incomparable with the class of Church-Rosser languages.
\end{lem} 
\bibliographystyle{eptcs}

\end{document}